\documentclass[parskip]{scrartcl}

\usepackage{amsmath,amssymb,amsthm}
\usepackage{bbm}
\usepackage{hyperref}

\newtheorem{theorem}{Theorem}[section]
\newtheorem{proposition}[theorem]{Proposition}
\newtheorem{lemma}[theorem]{Lemma}
\newtheorem{corollary}[theorem]{Corollary}

\newcommand{\hilbert}{{\mathcal H}}      
\newcommand{\fock}{{\mathcal F}}         

\newcommand{\N}{\mathbbm{N}}

\newcommand{\R}{\mathbbm{R}}
\newcommand{\C}{\mathbbm{C}}

\newcommand{\id}{\mathbbm{1}}            

\DeclareMathOperator{\spann}{span}       
\DeclareMathOperator{\tr}{tr}            

\begin{document}
\title{Boundedness properties of fermionic operators}
\author{Peter Otte \\ Fakult\"at f\"ur Mathematik \\ Ruhr-Universit\"at Bochum\\
Germany}
\maketitle
\begin{abstract}
The fermionic second quantization operator $d\Gamma(B)$ is shown to be
bounded by a power $N^{s/2}$ of the number operator $N$ given that the operator
$B$ belongs to the $r$-th von Neumann-Schatten class, $s=2(r-1)/r$. Conversely,
number operator estimates for $d\Gamma(B)$ imply von Neumann-Schatten conditions
on $B$. Quadratic creation and annihilation operators are treated as well.
\end{abstract}

\section{Introduction\label{introduction}}
Operators that satisfy the canonical anti-commutation relations (CAR) are
necessarily bounded. One may therefore ask what can be said about more
complicated operators, say, quadratic expressions in creation and annihilation
operators. Perhaps the most prominent such operator is $d\Gamma(B)$, the functor
of second quantization.

Suppose, we are given a Fock representation of the CAR over a separable complex
Hilbert space $L$. With the usual annihilation and creation operators $a(f)$ and
$a^\dagger(f)$ we define for a bounded operator $B$ on $L$ its second
quantization through
\begin{equation}\label{intro_dGamma}
  d\Gamma(B):=\sum_j a^\dagger(Be_j)a(\bar e_j)
\end{equation}
where $\{e_j\}$ is a complete orthonormal system (ONS) in $L$. The details of
this construction are briefly described in Section \ref{car}. We want to
compare $d\Gamma(B)$ with the number operator
$$
  N := d\Gamma(\id) =\sum_j a^\dagger(e_j)a(\bar e_j)
$$
There are two types of theorems. The first say, roughly, the more bounded $B$ is
the smaller $d\Gamma(B)$ is. More precisely, Theorem \ref{dGamma} tells us
\begin{equation}\label{intro_estimate}
  d\Gamma(B)^*d\Gamma(B) \leq 
\begin{cases}
  \|B\|_r^2 N^s + \|B\|_2^2\id  &  1< r< 2\\
  \|B\|_r^2 N^s                 &  r=1,\ 2\leq r\leq\infty
\end{cases}
\end{equation}
whenever $B$ is in the von Neumann-Schatten class $B_r(L)$, $1\leq r\leq
\infty$, and $s=\frac{2(r-1)}{r}$. The proof is based upon a thorough analysis of
\eqref{intro_dGamma} and uses H\"older and Cauchy-Schwarz inequalities for
operators. The literature provides estimates as in \eqref{intro_estimate} only
for the special cases $s=0$ ($r=1$) and $s=2$ ($r=\infty$). See Carey and
Ruijsenaars \cite{CareyRuijsenaars1987} and Grosse and Langmann
\cite{GrosseLangmann1992}. The $N^2$ bound holds also for bosons and looks
like what one would intuitively expect, namely, bound a quadratic operator by
another quadratic operator. However, thanks to the fermionic character, the
estimates can be improved upon to yield results for all $0\leq s\leq 2$.

In the second part, Theorem \ref{converse_dGamma} answers the question as to how
boundedness properties of $d\Gamma(B)$ affect the corresponding operator $B$
which is only interesting for $\dim L=\infty$. Its proof uses only elementary
calculations. For $s>0$ it turns out that in a way the bound
\eqref{intro_estimate} is almost sharp. That is to say, an estimate with $N^s$
implies $B\in B_{r+\varepsilon}(L)$ for all $\varepsilon>0$. For $s=0$ we may
even forget about $\varepsilon$ in that an estimate with $s=0$ implies $B\in B_1(L)$
which was conjectured by Carey and Ruijsenaars
\cite{CareyRuijsenaars1987} and Ottesen \cite{Ottesen1995}. It is an open
question whether one may drop $\varepsilon$ altogether.

All theorems proved for $d\Gamma(B)$ have analogs for the quadratic
annihilation and creation operators
\begin{equation}\label{intro_DeltaDeltaPlus}
  \Delta(A):=\sum_j a(Ae_j)a(\bar e_j),\
  \Delta^+(C):=\sum_j a^\dagger(Ce_j)a^\dagger(\bar e_j)
\end{equation}
Theorems \ref{Delta} and \ref{DeltaPlus} present number operator estimates in
the spirit of \eqref{intro_estimate} for $1\leq r\leq 2$ since $\Delta(A)$ and
$\Delta^+(C)$ are well-defined only for $A,C\in B_2(L)$. Hence, the $N^2$
estimates from the literature, see \eqref{literature_DeltaDeltaPlus}, are far
from optimal. The proofs parallel that for $d\Gamma(B)$. Contrary to
that, the converse Theorems \ref{converse_Delta} and \ref{converse_DeltaPlus}
are not elementary but employ a determinant formula for fermionic Gaussians and
a theorem from complex analysis. Their statement is essentially the same as for
$d\Gamma(B)$ except for the case $r=1$ which also has an $\varepsilon>0$.

\section{The CAR and second quantization\label{car}}
We sketch the necessary background from fermionic Fock space
theory. Presentations similar in spirit can be found in
\cite{CareyRuijsenaars1987} and \cite{Ottesen1995}. We formulate the CAR for
operator-valued functionals. To this end, let $L$ be a complex Hilbert space
equipped with a conjugation $f\mapsto \bar f$. Throughout, we will assume $L$ to
be separable. Let further $\fock$ be another complex Hilbert space. We call a
linear map from $L$ into the linear operators on $\fock$
$$
  f\in L,\ f\mapsto c(f)
$$
an operator-valued functional. The CAR need two such functionals, $a$ and
$a^\dagger$, which are assumed to have a common dense domain of definition
$D\subset\fock$ and
$$
  a(f)D\subset D,\ a^\dagger(f)D\subset D
$$
These operators are said to give a representation of the CAR if for all $f,g\in
L$ on $D$
\begin{gather}
  \{a(f),a(g)\} = 0 = \{a^\dagger(f),a^\dagger(g)\} \label{car1} \\
  \{a(f),a^\dagger(g)\} = (\bar f, g)\id \label{car2}
\end{gather}
where the curly brackets denote the anti-commutator. We further require the
unitarity condition
\begin{equation}\label{unitarity}
  a(f)^* = a^\dagger(\bar f)
\end{equation}
Properties \eqref{car1} through \eqref{unitarity} imply
\begin{equation}\label{projection}
  (a^\dagger(f)a(\bar f))^2 = \|f\|^2 a^\dagger(f)a(\bar f) \ \text{and}\
  (a^\dagger(f)a(\bar f))^* = a^\dagger(f)a(\bar f)
\end{equation}
In particular, $a^\dagger(f)a(\bar f)$ is an orthogonal projection for
$\|f\|=1$ and thus
\begin{equation}\label{projection_boundedness}
  0\leq a(f)^*a(f)\leq \|f\|^2,\ 0 \leq a^\dagger(f)^*a^\dagger(f)\leq \|f\|^2
\end{equation}
We have the fundamental boundedness result.

\begin{theorem}\label{boundedness}
The operators $a(f)$ and $a^\dagger(f)$ are bounded on their domain of
definition and therefore extend to bounded operators on all of $\fock$. We have
\begin{equation}\label{boundedness1}
  \|a(f)\| = \|a^\dagger(f)\| = \|f\|
\end{equation}
Hence, the maps $f\mapsto a(f)$, $f\mapsto a^\dagger(f)$ are continuous and
injective.
\end{theorem}

In what follows, we will work exclusively within the Fock representation. It
features a special vector, the vacuum $\Omega\in\fock$, $\|\Omega\|=1$. It is
annihilated by the $a(f)$'s
\begin{equation}\label{vacuum}
  a(f)\Omega = 0\ \text{for all}\ f\in L
\end{equation}
and cyclic for the $a^\dagger(f)$'s, i.e. 
\begin{equation}\label{cyclicity}
  \overline{\spann\{a^\dagger(f_{j_n})\cdots a^\dagger(f_{j_1})\Omega\mid
  n\in\N_0\}} = \fock
\end{equation}
Consequently, $a(f)$ is called annihilation operator and $a^\dagger(f)$ creation
operator. $\fock$ is the Fock space. Because of the vacuum the Fock space has a
special structure which can be described best through the $n$-particle spaces
\begin{equation}\label{n_particle_space}
  \fock^{(n)} := \overline{\spann\{ a^\dagger(f_n)\cdots
  a^\dagger(f_1)\Omega\}},\   n\geq 0
\end{equation}
It is clear that $\fock$ is built from these subspaces.

\begin{theorem}\label{fock_space}
The Fock space $\fock$ is the (completed) orthogonal sum of the $n$-particle
spaces $\fock^{(n)}$
$$
  \fock = \bigoplus_{n=0}^\infty\fock^{(n)}\ \text{with}\
  \fock^{(m)}\perp\fock^{(n)},\ m\neq n
$$
\end{theorem}

In order to avoid running into technical difficulties we will perform all
calculations on the subspace of finite particle numbers
\begin{equation}\label{F0}
  \fock_0 := \spann\{ \Phi \mid \Phi\in\fock^{(n)},\ n\in\N_0\}
\end{equation}
Creation and annihilation operators are fully understood by Theorem
\ref{boundedness}. The next more complicated operators are quadratic expressions
in creators and annihilators. Such quadratic operators are used in second
quantization as well as in constructing central extensions of certain Lie
algebras. There are different methods of introducing them. Here we define them
quite straightforwardly via the following series
\begin{gather}
  d\Gamma(B)  := \sum_j a^\dagger(Be_j)a(\bar e_j) \label{series_dGamma} \\
  \Delta(A)   := \sum_j a(Ae_j)a(\bar e_j),\ 
  \Delta^+(C) := \sum_j a^\dagger(Ce_j)a^\dagger(\bar e_j) \label{series_DeltaDeltaPlus}
\end{gather}
where $\{e_j\}$ is a complete ONS in $L$ and $A,B,C$ are linear operators on
$L$. The operator $d\Gamma(B)$ gives the functor of second quantization. When
$\dim L<\infty$ there is no problem of convergence. For general separable $L$
well-definedness can be shown under certain conditions at least on $\fock_0$.

\begin{theorem}\label{well-definedness}
Let $B:L\to L$ be bounded. Then, $d\Gamma(B)$ from
\eqref{series_dGamma} is well-defined on $\fock_0$ and
$d\Gamma(B)^*=d\Gamma(B^*)$. When $B\geq 0$ so is $d\Gamma(B)\geq 0$.
Furthermore, let $A,C:L\to L$ be Hilbert-Schmidt operators with $A^T=-A$ and
$C^T=-C$ where $A^T:=\bar A^*$ is the transpose. Then, $\Delta(A)$ and
$\Delta^+(C)$ from \eqref{series_DeltaDeltaPlus} are well-defined on $\fock_0$
and satisfy $\Delta(A)^*=\Delta^+(A^*)$.
\end{theorem}

We will not touch upon the question as to whether the domain of definition can
be enlarged. However, the conditions imposed on $A,B,C$ are in a way necessary.
For $d\Gamma(B)$ to exist on the entire one-particle space $\fock^{(1)}$ it is
necessary that $B$ is bounded. Likewise, in order that $\Delta(A)$ exists on the
entire two-particle subspace $\fock^{(2)}$ it is necessary that $A$ is
Hilbert-Schmidt. And finally, $\Delta^+(C)$ is defined on the vacuum only if $C$
is Hilbert-Schmidt.

We will need to know what $d\Gamma(B)$, $\Delta(A)$, and
$\Delta^+(C)$ do with the $n$-particle spaces
\begin{equation}\label{mapping_properties}
  d\Gamma(B):\fock^{(n)}\to\fock^{(n)},\
  \Delta(A):\fock^{(n)}\to\fock^{(n-2)},\
  \Delta^+(C):\fock^{(n)}\to\fock^{(n+2)}
\end{equation}
That is why $\Delta(A)$ and $\Delta^+(C)$ are called quadratic annihilation and
creation operators, respectively. $d\Gamma(B)$ preserves the number of
particles. Of all the interesting algebraic properties we only need one
commutator
\begin{equation}\label{comAC}
  [\Delta(A),\Delta^+(C)]   = -4d\Gamma(CA) + 2\tr AC\cdot\id
\end{equation}
By taking $B=\id$ we obtain the particle number operator or number operator for
short
$$
  N:=d\Gamma(\id) = \sum_j a^\dagger(e_j)a(\bar e_j)
$$
We will use the commutators
$$
  [N,a(f)] = -a(f),\ [N,a^\dagger(f)] = a^\dagger(f)
$$
As an operator on the Fock space $N$ has a very simple structure
\begin{equation}\label{restriction_number_operator}
  N\Phi = n\Phi,\ \Phi\in\fock^{(n)}
\end{equation}
which justifies the naming. Moreover, $N$ is essentially self-adjoint on
$\fock_0$ and $N\geq 0$. Since $N$ as well as its functions are just multiples
of the identity operator on each $\fock^{(n)}$ they commute with number
preserving operators.

\section{Number operator estimates\label{estimates}}
We want to estimate $d\Gamma(B)$, $\Delta(A)$, and $\Delta^+(C)$ by the number
operator $N$. The proofs usually rely on manipulating series, which are infinite
when $\dim L=\infty$. This can always be justified by standard arguments based
upon partial sums. For the sake of the presentation's clarity we will not carry
this out. Furthermore, we write $B_r(L)$ for the $r$-th von Neumann-Schatten
class and $B_r^-(L)$ for the subset of skew-symmetric operators $A^T=-A$.
Finally, for $1\leq r<\infty$ we will employ the singular value decomposition
\begin{equation}\label{singular_value}
  A = \sum_j \mu_j (e_j, \cdot)f_j
\end{equation} 
with singular values $\mu_j\geq 0$ and ONS's $\{e_j\}$ and $\{f_j\}$. When not
explicitly referring to \eqref{singular_value} we mean $\{e_j\}$ to be any ONS.

To begin with, we cite a Jensen type inequality for operators. It goes back to
Bhagwat and Subramanian \cite{BhagwatSubramanian1978}. See also
\cite{VasudevaSingh2008} and \cite{MondPevcaric1993}.

\begin{proposition}\label{jensen}
Let $w_j\in\R$, $w_j\geq 0$ for $j=1,\ldots,n$. Furthermore, let
$c_j:\hilbert\to\hilbert$ be bounded non-negative operators on a Hilbert space
$\hilbert$. Then, for all $1\leq p \leq q<\infty$
$$
   \Big( \sum_{j=1}^n w_j c_j^p\Big)^{\frac{1}{p}}
     \leq w^{\frac{1}{p}-\frac{1}{q}}
          \Big( \sum_{j=1}^n w_j c_j^q\Big)^{\frac{1}{q}}
$$
\end{proposition}

A simple consequence is a H\"older type inequality.

\begin{corollary}\label{hoelder}
Let $\mu_j\in \R$, $\mu_j\geq 0$, for $j=1,\ldots,n$. Let furthermore
$c_j:\hilbert\to\hilbert$ be bounded non-negative operators on a Hilbert space
$\hilbert$. Then, for $p,q\geq 1$, $\frac{1}{p}+\frac{1}{q}=1$
$$
  \sum_{j=1}^n \mu_j c_j
    \leq \Big(\sum_{j=1}^n \mu_j^p \Big)^{\frac{1}{p}}
          \Big(\sum_{j=1}^n c_j^q\Big)^{\frac{1}{q}}
$$
\end{corollary}
\begin{proof}
First of all, we rewrite the Jensen inequality in \ref{jensen} for a special case
$$
  \sum_{j=1}^n \mu_j c_j 
   \leq \Big(\sum_{j=1}^n \mu_j \Big)^{1-\frac{1}{q}} 
        \Big(\sum_{j=1}^n \mu_j c_j^q\Big)^{\frac{1}{q}}
$$
Without loss of generality we may assume $\mu_j>0$ for $j=1,\ldots,n$. Let
$\frac{1}{p}+\frac{1}{q}=1$. Then,
\begin{equation*}
  \sum_{j=1}^n \mu_j c_j
     = \sum_{j=1}^n \frac{\mu_j^p}{\mu_j^{p-1}} c_j 
     \leq \Big( \sum_{j=1}^n \mu_j^p\Big)^{1-\frac{1}{q}}
           \Big(\sum_{j=1}^n \frac{\mu_j^p}{\mu_j^{(p-1)q}}c_j^q\Big)^{\frac{1}{q}} 
     = \Big(\sum_{j=1}^n\mu_j^p\Big)^{\frac{1}{p}}
        \Big(\sum_{j=1}^n c_j^q\Big)^{\frac{1}{q}}
\end{equation*}
which is H\"older's inequality.
\end{proof}

This allows us to treat a very special case. 

\begin{lemma}\label{basic_estimate}
Let $\lambda_j\geq 0$. Assume
$$
  \Lambda_p := \Big(\sum_j \lambda_j^p\Big)^{\frac{1}{p}}<\infty,\ 
    \text{for}\ 1\leq p<\infty\ 
  \text{or}\
  \Lambda_\infty := \sup_j\lambda_j<\infty
$$
Then, for $\frac{1}{p}+\frac{1}{q}=1$ and with the understanding
$\frac{1}{\infty}=0$
$$
  \sum_j \lambda_j a^\dagger(e_j) a(\bar e_j)
    \leq \Lambda_p  N^{\frac{1}{q}}
$$
\end{lemma}
\begin{proof}
The simplest cases are $p=1,\infty$. For $p=1$,
$$
  \sum_j \lambda_j a^\dagger(e_j) a(\bar e_j)
    \leq \sum_j \lambda_j \id
$$
because of \eqref{projection_boundedness}. For $p=\infty$,
$$
  \sum_j \lambda_j a^\dagger(e_j) a(\bar e_j)
   \leq \sup_j \lambda_j \sum_j a^\dagger(e_j)a(\bar e_j).
$$
On to the cases $1<p<\infty$. By H\"older's inequality \ref{hoelder}
$$
  \sum_j \lambda_j a^\dagger(e_j) a(\bar e_j)
     \leq \Big(\sum_j \lambda_j^p\Big)^{\frac{1}{p}}
             \Big(\sum_j(a^\dagger(e_j)a(\bar e_j))^q\Big)^{\frac{1}{q}}
       =  \Big(\sum_j \lambda_j^p\Big)^{\frac{1}{p}} N^{\frac{1}{q}}
$$
since, by \eqref{projection}, $a^\dagger(e_j)a(\bar e_j)$ is an orthogonal
projection.
\end{proof}

At this point the fermionic character has entered the scene via
\eqref{projection_boundedness} and the calculations become invalid for bosons.
Lemma \ref{basic_estimate} can be applied to general operators by dint of an
operator version of Cauchy-Schwarz's inequality especially tailored to our
needs. Its proof mimics one of the elementary proofs.

\begin{proposition}\label{cauchy_schwarz}
Let $a_j,b_j:\hilbert\to\hilbert$ be bounded operators on a Hilbert space
$\hilbert$. Then, for $\sigma\in\{-1,1\}$
$$
  \sigma\sum_{j,k=1}^M a_j^*b_k^*b_ja_k
    \leq\sum_{j,k=1}^M a_j^*b_k^*b_ka_j
$$
\end{proposition}
\begin{proof}
Just look at the difference of both sides:
\begin{equation*}
\begin{split}
  2\sum_{j,k}(a_j^*b_k^*b_ka_j - \sigma a_j^*b_k^*b_ja_k)
    & = 2\sum_{j,k} a_j^*b_k^* ( b_ka_j-\sigma b_ja_k)\\
    & = \sum_{j,k} a_j^*b_k^* 
          (\sigma^2 b_ka_j -\sigma b_ja_k)
        +\sum_{j,k} a_k^*b_j^*(b_ja_k-\sigma b_ka_j) \\
    & = \sum_{j,k}\left( \sigma a_j^*b_k^*
          (\sigma b_ka_j-b_ja_k) + a_k^*b_j^*(b_ja_k-\sigma b_ka_j)
                             \right) \\
    & = \sum_{j,k} (\sigma a_j^*b_k^*-a_k^*b_j^*)
             (\sigma b_ka_j-b_ja_k) \\
    & = \sum_{j,k} (\sigma b_ka_j-b_ja_k)^*(\sigma b_ka_j-b_ja_k) \\
    & \geq 0
\end{split}
\end{equation*}
This implies the inequality.
\end{proof}

Now we can prove the first of the main theorems.

\begin{theorem}\label{dGamma}
Let $B\in B_r(L)$, $1\leq r\leq \infty$, and $s:=\frac{2(r-1)}{r}$. Then,
$$
  d\Gamma(B)^*d\Gamma(B) \leq 
\begin{cases}
  \|B\|_r^2 N^s + \|B\|_2^2\id  &  1< r< 2\\
  \|B\|_r^2 N^s                 &  r=1,\ 2\leq r\leq\infty
\end{cases}
$$
\end{theorem}
\begin{proof}
First of all, recall the singular value decomposition \eqref{singular_value}.
The simplest case $r=1$ follows immediately from
$$
  \Big\|\sum_j \mu_j a^\dagger(f_j)a(\bar e_j)\Big\|
    \leq \sum_j |\mu_j| = \|B\|_1
$$
On to the other cases. By the Cauchy-Schwarz inequality \ref{cauchy_schwarz},
\begin{equation*}
\begin{split}
  d\Gamma(B)^*d\Gamma(B)
    & = \sum_{j,k} a^\dagger(e_j)a(\overline{Be_j})a^\dagger(Be_k)a(\bar e_k)\\
    & = -\sum_{j,k} a^\dagger(e_j)a^\dagger(Be_k)a(\overline{Be_j})a(\bar e_k)
          + \sum_{j,k} (Be_j,Be_k)a^\dagger(e_j)a(\bar e_k) \\
    & \leq \sum_{j,k} \frac{\gamma_j^2}{\gamma_k^2}
          a^\dagger(e_j)a^\dagger(Be_k)a(\overline{Be_k})a(\bar e_j)
          + \sum_{j,k} (Be_j,Be_k)a^\dagger(e_j)a(\bar e_k) \\
    & =: \Sigma_2 + \Sigma_1
\end{split}
\end{equation*}
where $\gamma_j\in\R$, $\gamma_j\neq 0$, to be chosen appropriately. 

Let $1 < r < 2$. By dint of \eqref{singular_value} and Lemma
\ref{basic_estimate},
\begin{equation*}
  \Sigma_2  
      = \sum_k \frac{\mu_k^2}{\gamma_k^2} \sum_j \gamma_j^2
          a^\dagger(e_j)a(\bar e_j)
      \leq \sum_k \frac{\mu_k^2}{\gamma_k^2}
           \Big( \sum_j\gamma_j^{2p}\Big)^{\frac{1}{p}} N^{\frac{1}{q}}
\end{equation*}
with $\frac{1}{p}+\frac{1}{q}=1$. Upon choosing $\gamma_k=\mu_k^\alpha$ we
obtain
\begin{equation*}
    \Sigma_2 
       \leq \sum_k \mu_k^{2(1-\alpha)} \Big(\sum_j\mu_j^{2\alpha p}\Big)^{\frac{1}{p}} 
                N^{\frac{1}{q}}
\end{equation*}
We want $2(1-\alpha)=r$ and $2\alpha p=r$ which implies
$$
  \alpha=1-\frac{r}{2},\ p=\frac{r}{2-r}
$$ 
with $1<p<\infty$. Then,
$$
  \Sigma_2 \leq \Big(\sum_j \mu_j^r \Big)^{\frac{2}{r}} N^{\frac{2(r-1)}{r}}
$$
after some calculations. The sum $\Sigma_1$ can be estimated by
$$
  \Sigma_1 = \sum_j \mu_j^2 a^\dagger(e_j)a(\bar e_j)
    \leq \sum_j \mu_j^2 \id = \|B\|_2^2\id
$$
where the right-hand side is well-defined since $\|B\|_2\leq\|B\|_r$
for $1\leq r\leq 2$.
     
For $2\leq r<\infty$ we put $\gamma_j=1$ and use a different order of the
factors in $\Sigma_2$
\begin{equation*}
\begin{split}
  \Sigma_2 
    & =  \sum_{j,k} \mu_k^2 a^\dagger(f_k) a^\dagger(e_j)a(\bar e_j)a(\bar f_k) \\
    & \leq \sum_k \mu_k^2 a^\dagger(f_k) N a(\bar f_k) \\
    & = \sum_k \mu_k^2 a^\dagger(f_k)a(\bar f_k) N 
        - \sum_k \mu_k^2 a^\dagger(f_k)a(\bar f_k) \\
    & = N^{\frac{1}{2}} \sum_k \mu_k^2 a^\dagger(f_k)a(\bar f_k) N^{\frac{1}{2}} 
          - \Sigma_1
\end{split}
\end{equation*}
where we used that $N^{\frac{1}{2}}$ commutes with number preserving
operators. By Lemma \ref{basic_estimate},
$$
  \Sigma_2 \leq \|B\|_r^2 N^{\frac{r-2}{r}+1} -\Sigma_1
$$
which proves the present case.

The case $r=\infty$ needs a bit more care since we do not avail of a singular
value decomposition beforehand. Therefore, we look at the partial sums
$$
  d\Gamma_M(B) = \sum_{j=1}^M a^\dagger(Be_j)a(\bar e_j)
$$
The finite dimensional restriction
$$
  B_M := B \mid \spann\{e_1,\ldots, e_M\}
$$
however does have a singular value decomposition, the singular values
$\mu_j^{(M)}$ satisfying $\mu_j^{(M)}\leq \|B\|=\|B^*\|$ by the min-max
principle. Therefore, we can prove
\begin{align*}
  \sum_{j,k=1}^M a^\dagger(e_j)a^\dagger(Be_k)a(\overline{Be_k})a(\bar e_j)
       & \leq \|B\|^2\sum_{j=1}^M a^\dagger(e_j) N a(\bar e_j) \\
  \sum_{j,k=1}^M (Be_j,Be_k)a^\dagger(e_j)a(\bar e_k) 
       & \leq \|B\|^2 \sum_{j=1}^Ma^\dagger(e_j)a(\bar e_j)
\end{align*}
Thus,
\begin{equation*}
\begin{split}
  d\Gamma_M(B)^*d\Gamma_M(B)
    & \leq \|B\|^2 \sum_{j=1}^M a^\dagger(e_j)Na(\bar e_j) 
         + \|B\|^2\sum_{j=1}^Ma^\dagger(e_j)a(\bar e_j)  \\
    & = \|B\|^2 N^{\frac{1}{2}}\sum_{j=1}^M a^\dagger(e_j)a(\bar e_j)
             N^{\frac{1}{2}} \\
    & \leq \|B\|^2 N^2
\end{split}
\end{equation*}
That completes the proof.
\end{proof}

Now we turn to $\Delta(A)$ and $\Delta^+(C)$. Recall, that $A$ and $C$ must be
Hilbert-Schmidt operators for $\Delta(A)$ and $\Delta^+(C)$ to be well-defined
whence the following theorems only make sense for $1\leq r\leq 2$. Since Theorem
\ref{dGamma} contains the underlying ideas and computational details we may be
rather sketchy with the proofs.

\begin{theorem}\label{Delta}
Let $A\in B_r^-(L)$, $1\leq r\leq 2$, and $s:=\frac{2(r-1)}{r}$. Then,
$$
  \Delta(A)^* \Delta(A) \leq
\begin{cases}
  \|A\|_1^2\id      & r=1 \\
  \|A\|_r^2 N^s + \|A\|_2^2\id & 1<r\leq 2
\end{cases}
$$
\end{theorem}
\begin{proof}
We use the singular value decomposition \eqref{singular_value}. The case $r=1$
is obvious. For $1<r\leq 2$ we start, as in Theorem \ref{dGamma}, from
\begin{equation*}
\begin{split}
  \Delta(A)^*\Delta(A)
    & = \sum_{j,k} \mu_j\mu_k a^\dagger(e_j)a^\dagger(\bar f_j)a(f_k)a(\bar e_k)\\
    & = -\sum_{j,k}\mu_j\mu_k  a^\dagger(e_j)a(f_k)a^\dagger(\bar f_j)a(\bar e_k) 
           + \sum_j \mu_j^2 a^\dagger(e_j)a(\bar e_j)
\end{split}
\end{equation*}
For $1<r<2$ the proof runs along the same lines as in Theorem \ref{dGamma}.
However, for $r=2$ Cauchy-Schwarz's inequality \ref{cauchy_schwarz} gives us
\begin{equation*}
  \Delta(A)^*\Delta(A)
      \leq \sum_{j,k} \mu_k^2 a^\dagger(e_j)a(f_k)a^\dagger(\bar f_k)a(e_j) 
             + \sum_{j=1}^M \mu_j^2 a^\dagger(e_j)a(\bar e_j) 
      \leq \|A\|_2^2 N + \|A\|_2^2 \id
\end{equation*}
That completes the proof.
\end{proof}

The remaining operator $\Delta^+(C)$ could be treated in like manner. However,
it might be insightful to use an alternative idea. Note, that generally an
estimate for an operator does not yield an estimate for its adjoint.

\begin{theorem}\label{DeltaPlus}
Let $C\in B_r^-(L)$, $1\leq r\leq 2$, and $s:=\frac{2(r-1)}{r}$. Then,
$$
  \Delta^+(C)^*\Delta^+(C) \leq
\begin{cases}
  \|C\|_1^2 \id & r=1 \\
  \|C\|_r^2 N^s + 3\|C\|_2^2\id & 1<r\leq 2
\end{cases}
$$
\end{theorem}
\begin{proof}
The case $r=1$ is obvious. For $1<r\leq 2$ we use the commutator
$[\Delta,\Delta^+]$ from \eqref{comAC} to obtain
\begin{equation*}
\begin{split}
  \Delta^+(C)^*\Delta^+(C)
    & = \Delta(C^*)\Delta^+(C) \\
    & = \Delta^+(C)\Delta(C^*) + [\Delta(C^*),\Delta^+(C)] \\
    & =  \Delta(C^*)^*\Delta(C^*) - 4d\Gamma(CC^*) + 2\tr C^*C\cdot\id
\end{split}
\end{equation*}
Now use $d\Gamma(CC^*)\geq 0$ and Theorem \ref{Delta} to complete the proof.
\end{proof}

By using directly the defining series one could obtain better estimates, e.g.
for $r=2$
$$
    \Delta^+(C)^* \Delta^+(C) \leq  \|C\|_2^2 (N+2\id)
$$
It is instructive to write down the concrete bounds from the literature alluded
to in the introduction. Carey and Ruijsenaars have
\cite[2.14, 2.24, 2.25]{CareyRuijsenaars1987},
\begin{gather}
  d\Gamma(B)^*d\Gamma(B) \leq \|B\|_\infty N^2 \label{literature_dGamma}\\
  \Delta(A)^*\Delta(A) \leq \|A\|_2^2 N^2 ,\ 
  \Delta^+(C)^* \Delta^+(C) \leq \|C\|_2^2 (N+2\id)^2 \label{literature_DeltaDeltaPlus}
\end{gather}
When we assume $B$ just to be bounded, which is possible, then the estimate 
\eqref{literature_dGamma} for
$d\Gamma(B)$ is optimal. However, since $\Delta(A)$ and $\Delta^+(C)$ require
$A$ and $C$ to be Hilbert-Schmidt operators rather than bounded operators
\eqref{literature_DeltaDeltaPlus} does not give the correct magnitude at all.

The estimates by Grosse and Langmann \cite[App. B (b), (d)]{GrosseLangmann1992}
are derived in a super-version of the CCR and CAR. Being valid for bosons and
fermions alike they cannot reflect the special fermionic features used herein.

\section{Converse theorems\label{converse}}
Having seen Theorems \ref{dGamma}, \ref{Delta}, and \ref{DeltaPlus} one would
first and foremost ask whether the bounds given there are sharp. Since this is
not really interesting for $\dim L<\infty$ we tacitly assume $\dim L=\infty$. We
start with $d\Gamma(B)$ as this is the case which can be treated by elementary
means. The following statement for $r=1$ is also mentioned, without proof, in
\cite[p.7]{CareyRuijsenaars1987}.

\begin{theorem}\label{converse_dGamma}
Let $B\in B_\infty(L)$ and $d\Gamma(B)$ satisfy
\begin{equation}\label{converse_dGamma1}
  d\Gamma(B)^*d\Gamma(B) \leq \gamma_r N^s + \delta_r\id,\ 
  s=\frac{2(r-1)}{r}, \ 1\leq r<\infty
\end{equation}
Then $B\in B_1(L)$ for $s=0$. When $0<s<2$ then $B\in B_{r+\varepsilon}(L)$ for
all $\varepsilon>0$.
\end{theorem}
\begin{proof}
Let $\{e_j\}$ be any ONS. We start with the formula
\begin{equation}\label{converse_dGamma2}
  (a^\dagger(e_n)\cdots a^\dagger(e_1)\Omega,d\Gamma(B)a^\dagger(e_n)\cdots a^\dagger(e_1)\Omega)
    = \sum_{j=1}^n (e_j,Be_j)
\end{equation}
which along with the bound \eqref{converse_dGamma1} implies
\begin{equation}\label{converse_dGamma3}
  \Big|\sum_{j=1}^n (e_j,Be_j)\Big| \leq ( \gamma_r n^s + \delta_r)^{\frac{1}{2}}
\end{equation}
At first, we consider the special case of self-adjoint $B$. Then, either
$(e_j,Be_j)\geq 0$ or $(e_j,Be_j)<0$. For the ONS at hand we may permute the
indices as we wish without changing the right-hand side in
\eqref{converse_dGamma3}. Hence, with some constant $\gamma$
\begin{equation}\label{converse_dGamma4}
  \sum_{j=1}^n | (e_j,Be_j)| \leq \gamma n^{\frac{s}{2}}
\end{equation}
which in turn shows $(e_j,Be_j)\to 0$. If this were not so there would be an
$\varepsilon>0$ such that $|(e_j,Be_j)|\geq \varepsilon$ infinitely often. By
the permutation argument this would contradict \eqref{converse_dGamma4} since
$0\leq s<2$. Thus, we have shown that $(e_j,Be_j)\to 0$ for all ONS in $L$ which
implies $B$ is compact (see e.g. \cite{BakicGuljas1998}). Using in
\eqref{converse_dGamma4} the ONS from the singular value decomposition
\eqref{singular_value} we obtain
\begin{equation}\label{converse_dGamma5}
  \sum_{j=1}^n \mu_j \leq \gamma n^{\frac{s}{2}}
\end{equation}
where we noted $(e_j,f_j)=\pm 1$. For $s=0$ this implies $B\in B_1(L)$. Let
$s>0$. From \eqref{converse_dGamma5} we obtain the estimate
$$
  \mu_n \leq n^{\frac{s}{2}-1}
$$
For the powers $\mu_n^r$ to be summable it suffices that $r(1-\frac{s}{2})>1$.
This is equivalent to $\frac{2(r-1)}{r}>s$ which implies the statement for
self-adjoint $B$.

For general operators $B$ take real and imaginary parts in
\eqref{converse_dGamma2} and note $d\Gamma(B)^*=d\Gamma(B^*)$. Applying the
first part to $B+B^*$ and $i(B-B^*)$ completes the proof.
\end{proof}

For the operators $\Delta(A)$ and $\Delta^+(C)$ we need more machinery in
particular exponential functions of $\Delta^+(C)$. Fortunately, it is enough to
define them on the vacuum
$$
  \exp (z\Delta^+(C))\Omega,\ z\in\C
$$
where the exponential is defined via the power series. Such expressions were
studied by Robinson \cite{Robinson1995} and called fermionic Gaussians. In
physics one encounters the name BCS states. Their scalar product turns out to be
an entire analytic function in $z$.

\begin{lemma}\label{exponential_order}
Let $C\in B_2^-(L)$. Assume 
\begin{equation}\label{exponential_order1}
  \Delta^+(C)^*\Delta^+(C) \leq \gamma_r N^s + \delta_r\id,\ s:=\frac{2(r-1)}{r}
\end{equation}
for some $1\leq r\leq 2$. Then, the function
$$
  \omega(z) := (\exp(\bar z\Delta^+(C))\Omega, \exp( z\Delta^+(C)\Omega))
$$
is analytic on $\C$ and of exponential order $r$.
\end{lemma}
\begin{proof}
Recall from \eqref{mapping_properties} that
$\Delta^+(C):\fock^{(n)}\to\fock^{(n+2)}$ and 
$\fock^{(m)}\perp\fock^{(n)}$ for $m\neq n$. Then,
$$
  \omega(z) = \sum_{n=0}^\infty \frac{z^{2n}}{(n!)^2} 
     (\Delta^+(C)^n\Omega,\Delta^+(C)^n\Omega)
$$
Since the constants do not matter we may simplify the right-hand side of
\eqref{exponential_order1} to
$$
  \Delta^+(C)^*\Delta^+(C) \leq \gamma (N^s + \id)
$$
with $s=\frac{2(r-1)}{r}$ and some appropriate $\gamma$. Unfortunately, such
estimates do not transfer generally to powers of operators. Therefore, we have
to estimate by hand
$$
  (\Delta^+(C)^{n+1})^*\Delta^+(C)^{n+1}
    \leq \gamma (\Delta^+(C)^n)^* ( N^s + \id)\Delta^+(C)^n
$$
We know $\Delta^+(C)^n\Omega\in\fock^{(2n)}$ and
$N\mid\fock^{(2n)}=2n\id\mid\fock^{(2n)}$. Hence,
$$
  (\Omega,(\Delta^+(C)^{n+1})^*\Delta^+(C)^{n+1}\Omega)
     \leq \gamma ((2n)^s+1)(\Omega,(\Delta^+(C)^n)^*\Delta^+(C)^n\Omega)
$$
Successively,
\begin{equation*}
\begin{split}
  (\Omega,(\Delta^+(C)^{n+1})^*\Delta^+(C)^{n+1}\Omega)
     & \leq \gamma^{n+1} ((2n)^s+1)((2(n-1))^s+1)\cdots 1 \\
     & \leq \gamma^{n+1} 2^{n(s+1)}(n!)^s
\end{split}
\end{equation*} 
where the last estimate is for convenience. With an appropriate $\tilde z$,
$$
  |\omega(z)|
   \leq \sum_{n=0}^\infty \frac{|z|^{2n}\gamma^{n+1}2^{n(s+1)}}{(n!)^{2-s}}
    = \gamma \sum_{n=0}^\infty \frac{\tilde z^{2n}}{(n!)^{\frac{2}{r}}}
$$
since $2-s=\frac{2}{r}$. This shows $\omega$ is an entire function. Since
$1\leq r\leq 2$, we may use the classical Jensen inequality to deduce
$$
  |\omega(z)|
   \leq \delta \sum_{n=0}^\infty 
          \Big( \frac{\tilde z^{nr}}{n!}\Big)^{\frac{2}{r}}
   \leq \gamma \Big( \sum_{n=0}^\infty 
           \frac{\tilde z^{nr}}{n!}\Big)^{\frac{2}{r}}
   = \gamma \exp\Big( \frac{2}{r}\tilde z^r\Big)
$$
Hence, $\omega$ is of exponential order $r$.
\end{proof}

Lemma \ref{exponential_order} pertains to Fock space properties of
$\exp(z\Delta^+(C))$. On the other hand, we can express the scalar product on
$\fock$ through operators on $L$. See e.g. Robinson \cite{Robinson1995}.

\begin{proposition}\label{determinant}
Let $C\in B_2^-(L)$ and $z\in\C$. Then,
$$
  (\exp(\bar z\Delta^+(C))\Omega, \exp(z\Delta^+(C))\Omega)
    = \det(\id+4z^2C^*C)
$$
\end{proposition}

Combining Lemma \ref{exponential_order} with the determinant in Proposition
\ref{determinant}, hopefully, will tell us something about $C$. To this end, we
use a corollary of Jensen's integral formula from complex analysis that relates
the distribution of zeros of entire functions with their exponential order. See
\cite{Favorov2008} for the statement and some refinements.

\begin{theorem}\label{converse_DeltaPlus}
Let $C\in B_2^-(L)$. If $\Delta^+(C)$ satisfies the estimate
\begin{equation}\label{converse_DeltaPlus1}
   \Delta^+(C)^*\Delta^+(C) \leq \gamma_r N^s + \delta_r\id, \ s=\frac{2(r-1)}{r},\
   1\leq r\leq 2
\end{equation}
then $C\in B_{r+\varepsilon}^-(L)$ for all $\varepsilon>0$.
\end{theorem}
\begin{proof}
We use the formula from Proposition \ref{determinant}
$$
  \omega(z) := (\exp(\bar z\Delta^+(C))\Omega, \exp( z\Delta^+(C)\Omega)) 
             = \det(\id + z^2 C^*C).           
$$
Lemma \ref{exponential_order} and \eqref{converse_DeltaPlus1} imply
$\omega$ has exponential order $r$. Because of Proposition \ref{determinant}
the zeros $z_j\neq 0$ of $\omega$ are given through the singular values $\mu_j$
of $C$
$$
  z_j = \pm \frac{i}{\mu_j}\ \text{for all}\ \mu_j\neq 0
$$
The theorem from complex analysis mentioned above tells us
$$
  2\sum_j \mu_j^\alpha = \sum_j \frac{1}{|z_j|^\alpha} < \infty
$$
for all $\alpha>r$. Hence, $C\in B_\alpha^-(L)$ for all $\alpha>r$.
\end{proof}

Theorem \ref{converse_DeltaPlus} can be used for $\Delta(A)$ by the same
reasoning as in Theorem \ref{DeltaPlus}.

\begin{theorem}\label{converse_Delta}
Let $A\in B_2^-(L)$. If $\Delta(A)$ satisfies the estimate
\begin{equation}\label{converse_Delta1}
   \Delta(A)^*\Delta(A) \leq \gamma_r N^s + \delta_r\id, \ s=\frac{2(r-1)}{r},\
   1\leq r\leq 2
\end{equation}
then $A\in B_{r+\varepsilon}^-(L)$ for all $\varepsilon>0$.
\end{theorem}
\begin{proof}
As in Theorem \ref{DeltaPlus} we obtain the estimate
$$
  \Delta^+(A^*)^* \Delta^+(A^*)
     \leq \gamma_r N^s + \delta_r\id + 2\tr A^*A\cdot \id
$$
Then, Theorem \ref{converse_DeltaPlus} yields the statement.
\end{proof}

Theorems \ref{converse_dGamma}, \ref{converse_Delta}, \ref{converse_DeltaPlus}
naturally make one come up with the question as to whether the $\varepsilon$
could be removed there. Except for one special case, $r=1$ in Theorem
\ref{dGamma}, this is an open problem. If we could get rid of $\varepsilon$ the
bounds in Section \ref{estimates} would become sharp, at least asymptotically.
That this is so was conjectured by Carey and Ruijsenaars
\cite{CareyRuijsenaars1987} and Ottesen \cite{Ottesen1995} for the case $r=1$.
Our proofs as they stand cannot be generalized. The estimate of the singular
values in Theorem \ref{converse_dGamma} is sharp as show simple examples. As to
Theorem \ref{converse_DeltaPlus} there are entire functions of exponential order
$1$ whose zeros cannot be summed up with exponent $1$, e.g. $f(z)=\sin(z)$.
Hence, although the operators
$$
  d\Gamma(B) = \sum_j \frac{1}{j} a^\dagger(f_j)a(\bar e_j),\
  \Delta^+(C) = \sum_j \frac{1}{j} a^\dagger(f_j)a^\dagger(\bar e_j)
$$
look quite similar we only know the first to be unbounded whereas the latter's
unboundedness remains an open problem.

\section*{Acknowledgment}
This work was supported by the research network SFB TR 12 -- `Symmetries and
Universality in Mesoscopic Systems' of the German Research Foundation (DFG).

\bibliographystyle{plain}
\bibliography{literature}
\end{document}